\documentclass[twocolumn,smallextended]{svjour3}          

\RequirePackage{amsmath,amsfonts}
\RequirePackage{hyperref}
\RequirePackage{natbib}
\RequirePackage[framemethod=default]{mdframed}
\usepackage{algorithmic}
\usepackage{algorithm}
\RequirePackage{fix-cm} 
\usepackage[T1]{fontenc}
\usepackage{graphicx}
\newcommand{\PTalgo}{Parallel Tempering algorithm}
\newcommand{\PT}{\textit{Parallel Tempering}}
\newcommand{\pt}{\textsc{PT}}

\newcommand{\walkKernel}{\mathcal{M}} 
\newcommand{\swapProposal}{\mathcal{Q}} 
 
\newcommand{\ind}[1]{\mathbf{1}_{\{ #1 \}}}

\newcommand{\real}{\mathbb{R}}
\newcommand{\bigSigmaAlgebra}{\boldsymbol{\mathfrak{F}}}
\newcommand{\swapKernel}{\mathcal{S}}
\newcommand{\swapKernelij}{\mathcal{S}_\mathit{ij}}
\newcommand{\bdpi}{\boldsymbol{\pi}}
\newcommand{\lebesgue}{\lambda_\text{Leb}}
\newcommand{\EV}{\mathbb{E}}

\newcommand{\pij}{p_\mathit{ij}}
\newcommand{\Tij}{T_\mathit{ij}}
\newcommand{\alphaij}{\alpha_\mathit{ij}}
\newcommand{\rmd}{\rm{d}}

\newcommand{\itx}{\mathit{x}}
\newcommand{\bdy}{\boldsymbol{y}}
\newcommand{\ity}{\mathit{y}}

\newcommand{\itA}{\mathit{A}}

\newcommand{\itB}{\mathit{B}}

\newcommand{\itAxB}{{\itA \times \itB}}

\newcommand{\itBxA}{{\itB \times \i}}
\usepackage{multirow}
\usepackage{caption}
\usepackage{booktabs}

\smartqed  
\journalname{ }

\begin{document}

\title{State dependent swap strategies and adaptive adjusting of number of temperatures in Parallel Tempering algorithms.
  \thanks{Work partially suported by Polish National Science Center grant No. N N201 608740 and grant No. 2011/01/B/NZ2/00864}
    }

\titlerunning{State dependent swap steps in adaptive PT}       

\author{Mateusz Krzysztof {\L}\k{a}cki \and
  B{\l}a{\.z}ej Miasojedow}  
\institute{
    B. Miasojedow \at
    Institute of Applied Mathematics, University of Warsaw\\ 
    Banacha 2, 02-097 Warsaw,  Poland\\
    Tel.: +48-22-55-44-441\\
    \email{B.Miasojedow@mimuw.edu.pl}
    \and
    M.K. {\L}\k{a}cki \at
    Institute of Informatics, University of Warsaw\\ 
    Banacha 2, 02-097 Warsaw,  Poland\\
    Tel.: +48-602-39-50-10\\
    \email{Mateusz.Lacki@biol.uw.edu.pl}
}

\date{Received: date / Accepted: date}

\maketitle

\begin{abstract}

In this paper we present extensions to the original adaptive \PT\,algorithm. Two different approaches are presented. In the first one we introduce state-dependent strategies using current information to perform a swap step. It encompasses a wide family of potential moves including the standard one and \textit{Equi Energy} type move, without any loss in tractability. In the second one, we introduce online adjustment of the number of temperatures. Numerical experiments demonstrate the effectiveness of the proposed method.
\end{abstract}

\section{Introduction}\label{introduction}

Markov chain Monte Carlo (MCMC) is a generic method to approximate an integral of the form
\[
     I := \int_{\real^d} f(y) \pi(y) \rmd y \;,
\]
where $\pi$ is a probability density function, which can be evaluated point-wise up to a normalising constant. Such an integral occurs frequently when computing Bayesian posterior expectations \citep{robert-casella,gilks-mcmc}.

The random-walk Metropolis algorithm \citep{metropolis} often works well, provided the target density $\pi$ is, roughly speaking, sufficiently close to unimodal. The efficiency of the Metropolis algorithm can be optimised by a suitable choice of proposal distribution. These, in turn, can be chosen automatically by several adaptive MCMC algorithms; see \cite{haario-saksman-tamminen-am,atchade-rosenthal,roberts-rosenthal-examples,andrieu-thoms} and references therein.

When $\pi$ has multiple well-separated modes, the random-walk based methods tend to stuck in a single mode for long periods of time. It can lead to false convergence and severely erroneous results. Using a tailored Metropolis-Hastings algorithm can help, but, in many cases, finding a good proposal distribution is not easy. Tempering of $\pi$, that is, considering auxiliary distributions with density proportional to $\pi^\beta$ with $\beta\in(0,1)$, often provides better mixing between modes \citep{swendsen-wang,marinari-parisi,hansmann,woodard-schmidler-huber-rapid}. 

We focus here particularly on the parallel tempering algorithm, which is also known as the replica exchange Monte Carlo and the Metropolis coupled Markov chain Monte Carlo.


The tempering approach is particularly tempting in such settings where $\pi$ admits a physical interpretation, and there is good intuition how to choose the temperature schedule for the algorithm. 

In general, choosing the temperature schedule is a non-trivial task, but there are generic guidelines for temperature selection based on both empirical findings and theoretical analysis \citep{kofke,kone-kofke,atchade-roberts-rosenthal,roberts-rosenthal-simulated-tempering}. These theoretical findings were used to derive adaptive version of the \PT\, \citep{BM2}.
  
In the present paper we consider the adaptive version of the \PTalgo. The adaption consists in introducing state dependent swaps between differently tempered random walks. We study the impact of different distributions on potential steps and call them \textit{Strategies}. Our choice of strategies is driven by solutions already known to the literature \citep{kou2006} and used within \PT\, algorithm by \cite{ BaragattiParallelTemperingWithEquiEnergyMoves}. The novelty of our approach stems from an alternative implementation of \textit{Equi Energy} moves that renders the algorithm parameters free, i.e. the user does not need to provide precise \textit{Energy Rings} any more. We also investigate different modifications of this new approach.

We also propose an automated method for adapting the actual number of considered temperatures, in the spirit of \cite{BM2}. The temperature adaptation scheme depends on the parameters of the adaptive random walks applied in the parallelised Metropolis-Hastings stage of the algorithm in case when the state space amounts to be the usual $\real^d$. 

We have also showed that the proposed algorithm satisfies the Law of Large numbers, in the same setting as in \cite{BM2}.

\section{Definition and Notations}\label{definitionsNotation}

Our basic object of interest is the density $\pi: \Omega \mapsto \real_+$, where $\Omega = \real^d$. We assume we can evaluate pointwise a function that is proportional to $\pi$ by some constant. The \PT\, approach suggests to construct a Markov chain on the product space $\Omega^L$, where $L$ is the number of temperature levels. On that space a new density $\pi^\beta$ is constructed by posing for $\itx \in \Omega^L$
\[
	\pi^\beta ( \itx ) = \pi^\beta ( \itx_1, \dots, \itx_L) \propto \pi^{\beta_1}( \itx_1 ) \times\dots\times \pi^{\beta_L}( \itx_L )
\]
so that $\itx_i \in \Omega$, $\beta = \big( \beta_1, \beta_2, \dots, \beta_L\big)$ are the inverse temperatures subject to $1 = \beta_1 > \beta_2 > \dots > \beta_L$.

Vector $T = (T_1, \dots, T_L)$, where $ T_\ell = \beta_\ell^{-1}$, contains numbers known as temperatures and is itself referred to as the temperature scheme. 

Density $\pi^\beta$ is known up to proportionality factor and by marginalising it w.r. to the first coordinate we retrieve the original distribution $\pi$. 
	 
Markov chain $X = \{ X^{[k]}\}_{k\geq0}$ targets $\pi^\beta$ and can be spit into $L$ coordinate chains, since $X^{[k]} = \big( X^{[k]}_1, \dots X^{[k]}_L\big)$. First coordinate chain will be referred to as the base chain. 

The main idea behind \PT\, is to interweave random walk steps with random swaps between chains. Each random swap exchanges results of a random walk step from two coordinate chains. Chains corresponding to higher temperatures \footnote{Chains with lower $\beta$.} should in principle be more volatile and travel between different modes more easily than chains linked to lower temperatures. For if $\itx$ is the last visited place by the $l^\text{th}$ chain, and $\ity$ is a proposal drawn from a region where the density assumes smaller values, $\pi( \ity ) < \pi( \itx )$, then the probability of accepting such proposal, equal to $\eta_\ell$, is higher on the more tempered chain\footnote{Here we assume that random walk proposal kernel is symmetric.}  
\[ 
	\eta_\ell( \itx, \ity ) =  1 \wedge \Biggl(\frac{\pi( \ity )}{\pi( \itx )}\Biggl)^{\beta_l} 
	 > 1 \wedge \frac{\pi( \ity )}{\pi( \itx )} = \eta_1( \itx, \ity ).
\]
Therefore, the more exchanges of higher tempered chains with the base chain, the bigger the chance of escaping quickly local probability clusters where a simple Markov chain would usually get stuck.

Formally, the generation of Markov chain $X$ proceeds by repeated application of two kernels, 

\[
	X^{[n-1]} \overset{\walkKernel}{\longrightarrow} \widetilde{X}^{[n-1]} \overset{\swapKernel}{\longrightarrow} X^{[n]}.
\] 

The random walk kernel $\walkKernel$ is defined on the product sigma algebra generator by
\[
 \walkKernel( \itx, A_1 \times\dots\times A_L)  = \prod_{l=1}^L \walkKernel_l ( \itx_l, A_l ),
\]  
where 
\[
	\walkKernel_l (\itx_l, A_l) \equiv 
		\int_A \eta_\ell (\itx_l, \ity_l) q(\ity_l - \itx_l) \rmd\,\ity_\mathit{l} + \ind{ \itx_\mathit{l} \in \mathit{A}} r(\itx_\mathit{l}), 
\]
and corresponds to independently performing random walk on different chains.

Let $\Tij$ denote the transposition of coordinate $i$ and $j$ of  
\[ 
	\Tij( \itx ) = (\itx_1, \dots, \itx_j, \dots, \itx_i, \dots, \itx_L ).
\]

Then the proposal measure for the random swap is discrete and potentially state-dependent and is given by

\begin{equation}
	\swapProposal( \itx, \itA ) \equiv \sum_{i < j} \pij( \itx ) \ind{\itA}( \Tij \itx),
\end{equation}

where we one is free to choose the precise relation between $\pij$ and $\itx$, establishing with what probabilities should one carry out the transposition.

The transition kernel of the swap step is defined simply by
\begin{equation}\label{eq:swap-kernel}
    \begin{split}
        \swapKernel( \itx, \itA ) :=& \sum_{i<j} \ind{ \Tij( \itx ) \in \itA }  \pij( \itx ) \alphaij( \itx ) \\
          & + \ind{ \itx \in \itA } \sum_{i<j}  \Big( 1-\alphaij(\itx) \Big) \pij( \itx )\;.    
    \end{split}    
\end{equation}

The precise form of the acceptance probabilities will be presented in next section.

\section{State dependent swap strategies}\label{stateDependentSwapStrategies}

In standard Parallel Tempering algorithm the proposal of swap step is drawn from distribution independent on the current state of process. Usual choice is to sample pair of temperatures uniformly at random or propose only pair of adjacent temperatures chosen also uniformly. Motivation of studies strategies of swap dependent on current state is Equi-Energy sampler proposed by \cite{kou2006} and adopted to \PT\, algorithm in \cite{BaragattiParallelTemperingWithEquiEnergyMoves}. 

The main idea behind these algorithms is to exchange states with similar energy (i.e. value of log-density). The original \textit{Equi Energy} sampler is greedy in memory usage: it must store all points drawn from differently tempered trajectories. 
This might become problematic in high dimensional settings. \textit{Equi Energy} also targets a biased distribution instead of the proper one due to the use of asymptotic values in acceptance probability.
Running this algorithm also requires specifying in advance the so called Energy Rings, i.e. the partition of the state space into regions with similar energy. \cite{Schreck} addresses the problem of choosing Energy Rings and provides an and adaptive version of algorithm. 

The Parallel Tempering algorithm with Equi-Energy moves proposed by \cite{BaragattiParallelTemperingWithEquiEnergyMoves} also requires to specify Energy-Rings. This hindrance can be circumvented with the use of state dependent swap strategies performing Equi-Energy-like moves inside the Parallel Tempering algorithm with no need to predefine Energy Rings themselves. In addition, our approach is flexible and using different strategies one can promote large jumps or other features. 

The general algorithm is as follows: given current state of process after the random walk phase, $\itx = ( \itx_1, \dots, \itx_L )$, we propose that the drawing of transposition\footnote{We exchangeably refer to coordinate transposition as swaps.} $\Tij( \itx )$ should follow the state dependent discrete distribution defined by some probabilities $\pij(\itx)$, defined on a discrete simplex $i<j$. 

To assure reversibility, the swap acceptance probability should be defined by

\begin{equation}\label{eq:swap-acceptance}
    \alphaij(\itx) := 
    \frac{ \pij \big( \Tij( \itx ) \big) }{ \pij(\itx) } 
    \left( \frac{ \pi( \itx_i )}{ \pi( \itx_j ) } \right)^{ \beta_j-\beta_i }
    \wedge 1\;.
\end{equation}

The definition of acceptance probability \eqref{eq:swap-acceptance} assures that kernel $\swapKernel$ is reversible with respect to $\bdpi$.

\begin{proposition}
    Kernel $\swapKernel$ defined by \eqref{eq:swap-kernel} is reversible with respect to 
    $\bdpi( \itx ) \propto \pi(\itx_1)^{ \beta_1 } \times\cdots\times \pi( \itx_L )^{\beta_L} $.
\end{proposition}

\begin{proof}
    We need to show that for all $\itA, \itB \in \bigSigmaAlgebra$ we have
\[
    \int_\itAxB     \bdpi( \rmd\itx )   \swapKernel( \itx, \rmd \ity) =
    \int_\itAxB     \bdpi( \rmd\ity )   \swapKernel( \ity, \rmd \itx)
    \;.
\]

For all $i<j$ let us define 
\[
    \swapKernelij( \itx, \itA ) := \ind{ \Tij( \itx ) \in \itA } \pij( \itx ) \alphaij( \itx )\;.
\] 

It is enough to verify that for every $i<j$
\[
    \int_\itAxB     \bdpi( \rmd\itx )   \swapKernelij( \itx, \rmd\ity ) =
    \int_\itAxB     \bdpi( \rmd\ity )   \swapKernelij( \ity, \rmd\itx )\;.
\]

For any arbitrary chosen $i<j$ define a measure $\mu$ on $\real^{2dL}$ as follows: for $\itA, \itB \in \bigSigmaAlgebra$ let 
\[
    \mu( \itAxB )   := \lebesgue \Big( \big\{ \itx \in \itA \;:\; \Tij( \itx ) \in \itB ) \big\} \Big)\;,
\]
where $\lebesgue$ denotes the Lebesgue measure on $\real^{dL}$. 

Since $\Tij \big( \Tij( \itx ) \big) = \itx$, by symmetry of Lebesgue measure we get 
\begin{equation}\label{eq:symmetry-mu}
    \begin{split}
        \mu( \itAxB )   &= \lebesgue\Big( \itA \cap \Tij(\itB) \Big) = \\
                        &= \lebesgue\Big( \Tij(\itA) \cap \itB \Big) = \mu( \itBxA )\;,
    \end{split}
\end{equation}
and by definition of $\swapKernelij$ we obtain that
\begin{equation}\label{eq:proof_prop1}
    \begin{split}
        \int_\itAxB &\bdpi( \rmd\itx ) \swapKernelij( \itx, \rmd\ity )   =  \\
    &=  \int_\itAxB \bdpi( \itx ) \alphaij(\itx) \pij(\itx) \mu( \rmd\itx, \rmd\ity ) =\\ 
    &=  \int_{ \itA \cap \Tij(\itB) } \bdpi( \itx ) \alphaij(\itx) \pij(\itx) \rmd\itx             
    \end{split}
\end{equation}

Now, using \eqref{eq:swap-acceptance} we find that
\[
    \bdpi(\itx) \alphaij(\itx) \pij(\itx)   =  
    \bdpi\big( \Tij(\itx) \big) \alphaij\big( \Tij(\itx) \big) \pij\big( \Tij(\itx) \big)\;,
\]
and setting $\ity=\Tij(\itx)$ to \eqref{eq:proof_prop1} and applying \eqref{eq:symmetry-mu} we get
\[
    \begin{split}
        \int_\itAxB &\bdpi( \rmd\itx ) \swapKernelij( \itx, \rmd\ity )   =&  \\
    &=  \int_{ \Tij(\itA) \cap \itB } \bdpi( \itx ) \alphaij(\itx) \pij(\itx) \rmd\itx &= \\
    &=  \int_\itAxB \bdpi( \itx ) \alphaij(\itx) \pij(\itx) \mu( \rmd\ity, \rmd\itx ) &=\\           
    &=  \int_\itAxB \bdpi( \rmd\ity ) \swapKernelij( \ity, \rmd\itx ),  &\\
    \end{split}
\]
which completes the proof.
\end{proof}

\begin{remark}
 Thanks to the definition of kernel $\swapKernel$, for any positive measurable function $F:\real^{Ld}\to\real^+$ invariant by permutation we get $\swapKernel F(\itx) = F(\itx)$ which, under some regularity conditions \citep{BM2}, implies that \PT\, algorithm with state dependent swap steps is geometrically ergodic. For under the same assumptions \textsc{Theorem 1} precised in the above-mentioned \textit{\oe uvre} holds with state dependent swap steps. 
\end{remark}

To pursue an \textit{Equi Energy} type moves without a need to fine-tune some additional parameters, e.g. precising \textit{Energy Rings}, we propose to set the swap probabilities as follows: for $i<j$ let 
\begin{equation}\label{eq:def-strategy}
    \pij(\itx) \propto \exp\big\{ -| \log( \pi(\itx_i)-\log( \pi(\itx_j) | \big\}\;.
\end{equation}
The norming constant is equal $\sum_{i<j}\pij(\itx)$, so it does not depend upon permutation of $\itx$, thus $\pij(\itx)=\pij\big(\Tij(\itx)\big)$, and further acceptance probability simplifies to acceptance probability of standard Parallel Tempering algorithm
\begin{equation}\label{eq:acceptance-symetry}
    \alphaij(\itx) = \left( \frac{ \pi(\itx_i) }{ \pi(\itx_j) } \right)^{\beta_j-\beta_i}    \wedge 1\;.
\end{equation}

In comparison with standard PT, the state dependent swaps version propose more often states which are close in the sense of difference of energy. But, since acceptance mechanism is the same this leads to
increasing number of global moves in the algorithm. The simulations results, presented below, confirm this improvement.

Other possible strategies commonly used in the literature include swapping at random (RA) all possible pairs, thus setting $ \pij = { L \choose 2 }^{-1} $, and swapping at random only the adjacent levels (AL), where $\pij = (L-1)^{-1} \boldsymbol{1}_{\{i=j-1\}}$.  

\begin{remark}  
    Note, that PT with Equi-Energy moves can be considered a special case of state dependent swap step. Let $H_1,\dots,H_m$ be the Energy Rings. Denote by $H_x$ the set $H_i$ where $x$ belongs to. Setting
    $\pij(\itx) \propto \ind{ \itx_j \in H_{\itx_i} }$, the acceptance probabilities also reduce to \eqref{eq:acceptance-symetry}. In the end we obtain an algorithm proposed by \cite{ BaragattiParallelTemperingWithEquiEnergyMoves}.

    Therefore the theoretical results presented in the present paper can be directly applied to \textsc{PTEEM}. In particular, it is geometrically ergodic under the same regularity condition, also convergence and the Law of Large Numbers for \textsc{PTEEM} with adaptive Metropolis step at each level can be thus easily obtained.  
\end{remark}

\section{Adjusting number of temperature levels}\label{adjustingTemperatureLevels}

The choice of parameters, i.e. the schedule of temperatures and parameters of random walk Metropolis at each level, is crucial for performance of the \pt\, algorithm. This becomes apparent especially in case of the temperatures not assuming direct physical interpretation. 

In a recent paper \cite{BM2} has proposed an adaptive scheme to tune both temperature schedule and proposal distributions of the random walk steps at each temperature level. The above-mentioned algorithm still requires user-provided number of temperature levels and usually some pilot runs of the algorithm are necessary to determine their number.

In this section a simple criterion is presented for probing whether the algorithm has attained the correct number of temperature levels during its run-time. An incorrect temperature schedule leads to unsatisfactory estimation results: underestimating the number of temperature levels might lead the most-tempered chain to explore a distribution that is still multimodal, contradicting the whole idea of parallel tempering.
On the other hand, overestimating levels number might lead to other sort of inefficiencies. For, clearly, the waiting time for swaps between ground level chain and its more tempered counterparts will increase together with the overall number of temperatures. This is because if at each step we try to swap only two chains then the probability of proposing to swap the chain with original target with others decreases. It implies that the best candidate for maximal temperature is to take the smallest temperature at which the distribution becomes unimodal. Notice, that if more swaps occur between the tempered levels, then we do not really solve the problem of poor mixing at the ground level.

There are two possible solutions to the above-mentioned problem contingent upon our knowledge on $\pi$. If the number of modes is known in advance, one could simply observe the created trajectories and consider additional temperature levels whenever the algorithm did not arrive at finding some modes. Triming the temperature schedule could also be arranged if the distribution was known to be unimodal at more than one level. In general, the number of modes is unknown and such procedures would not be feasible, and the only safe solution remaining is to overestimate the number of levels and adaptively decrease it. That is precisely the approach we follow. In practice this leaves still one parameter to be provided by the user --- the initial number of temperatures. 


It is easy to construct toy problems with the temperatures extending over too limited a range. In such cases the \PT\,algorithm would not find other modes except the one where it was initialized.  In such cases nearly all  criterions to test the target multimodality would fail. To decide that on some level target distribution is unimodal or not we will use comparison of covariance matrices of proposal of random walk and of target.

The well known results on scaling of random walks Metro\-polis \citep{roberts-gelman-gilks-scaling} for i.i.d targets and further extensions for more general form of target distributions) choice of the optimal  covariance of proposal distribution.
This choice is to set proposal covariance equal $\frac{2.38^2}{d}\Sigma$, where $\Sigma$ is the covariance of the stationary distribution, which corresponds to stationary acceptance probability equal $0.234$.
There are two popular ways of adaptive schemes which tries to approximate this optimal choice. The first is to estimate on-line the stationary covariance by $\Sigma^{[n]}$ and in next step use proposal with a covariance $\frac{2.38^2}{d}\Sigma^{[n]}$ \citep{haario-saksman-tamminen-am}. A standard way to define $\Sigma^{[n]}$ follows a two step procedure\footnote{For simplicity we do omit the chain level index, $\ell$.}:
\begin{itemize}
    \item Get the estimator of the target's mean by 
    \begin{equation}\label{estimTargetMean}
        \mu^{[n]} = (1 - \gamma^{[n]}) \mu^{[n-1]} + \gamma^{[n]} X^{[n-1]}.      
    \end{equation}  
    \item Get the estimator of the target's covariance matrix
    \begin{multline}\label{estimTargetVar}
        \Sigma^{[n]} = (1 - \gamma^{[n]}) \Sigma^{[n-1]} +\\+ \gamma^{[n]} (X^{[n-1]} - \mu^{[n]})(X^{[n-1]} - \mu^{[n]})^\text{t}.
    \end{multline}
\end{itemize}
In \eqref{estimTargetMean} and \eqref{estimTargetVar} \par $\gamma^{[n]}$ denotes a deterministic step-size\footnote{A usual choice is $\gamma^{[n]} \propto n^{-\alpha}$, where $\alpha \in (0.5, 1)$.}.

The second one is to use proposal with covariance of the form $\exp(2\theta^{[n]})\Sigma^{[n]}$ , where $\theta^{[n]}$ is adjusted in order to get $0.234$ acceptance rate \citep{andrieu-thoms,roberts-rosenthal-examples} by the following a Robbins-Monro type procedure
\begin{equation}
    \theta^{[n]} = \theta^{[n-1]} + \gamma^{[n]} ( \eta^{[n-1]} - 0.234 ), 
\end{equation}
where $\eta^{[n-1]}$ denotes the acceptance probability of the random walk phase.

In the case of unimodal distribution both approaches are approximately equivalent. In the multimodal case it is obvious that $\exp(2\theta^{[n]})$ will be significantly lower than $\frac{2.38^2}{d}$. It is motivated by the need to shrink the covariance of proposal so that it mimics the covariance of the target distribution truncated to current mode, i.e. where the chain currently is at.    

The number of temperature levels is adjusted according to the following formula
\begin{equation}\label{eq:condition}
    L^{[n]} = \min \left\{ \ell \in \{1, \dots, L^{[n-1]} \}: 
        e^{\theta^{[n]}_\ell} \geq \frac{2.38}{\sqrt{d}}
      \right\},
\end{equation}

if the condition is satisfied for at least one $\ell$ and $L^{[n]} = L^{[n-1]}$, otherwise. 

It seems plausible to froze adaptation during initial burn-in period until the estimators of covariances $\Sigma^{[n]}_\ell$ and the scaling factors $e^{\theta^{[n]}_\ell}$ stabilise. 

\begin{remark}
    A similar criterion to that presented in \eqref{eq:condition} uses a different threshold, namely 

\[
    \min{\left\{ \frac{2.38}{\sqrt{d}}, 
    \underset{\ell = 1, \dots, L^{[n-1]}}{\max}{ 
        \left\{ 
            (1 - \epsilon)e^{\theta^{[n]}_\ell}
        \right    \} } \right\} 
    }.
\]

    The above condition is more lenient than the one presented in \eqref{eq:condition}, still being optimal when applied to the Gaussian target. In general $\frac{2.38^2}{d}\Sigma$ might be more than achieved using the optimal scaling, this being equal to $\frac{2.38^2}{I\,d}$ in case of a target density in product form, where $I$ denotes the Fisher information of the shift parameter \citep{roberts-gelman-gilks-scaling}. This stems from the Rao-Crammer lower bound, $I^{-1}~\leq~\Sigma$.

\end{remark}

\section{Algorithm}\label{algorithm}

We shall now pass to detailed description of our algorithm.

In the beginning we choose the initial positionings of all the chains arbitrary. We then apply iteratively the basic step of the \emph{Adaptive Parallel Tempering} as described in Algorithm \ref{algoDescription}. Each step depends on the outputs of the previous step through the positionings of chains in the state space, current estimator of the target covariance matrix and target expected value, the covariance scaling factor, and the details of the temperature scheme such as precise values of temperatures and their overall number. Also, a sequence of descending numbers $\gamma^{[n]}$ is being used.   

The basic step consists of five phases. 

In the first phase a simple random walk is carried out independently on every chain. The proposal for next step is drawn from the multivariate normal distribution with the input covariance matrix scaled by the $\exp(2\theta^{[n]})$. The proposals are being accepted or rejected following the standard procedure assuring the reversibility of the corresponding kernel.

In the second phase we perform random walk adaptation scheme. It consists in updating estimates of both the expected value and the covariance matrix of the target distribution. Also the proposal covariance scaling factor is being updated so that to make it dependent of the distance between the proposal acceptance calculated in the random walk phase and the theoretical optimal probability of acceptance, mentioned in Section \ref{adjustingTemperatureLevels}. Every update is done using a sequence of weights $\gamma^{[n+1]}$ constructed so as to give less and less attention to the sampled values while the algorithm proceeds. It is also worth mentioning that no explicit calculations are done using full covariance matrices. Instead, we operate on their Cholesky decompositions and update them using the so called \textit{rank-one update}, whose cost is quadratic with matrix dimension. Also, it is easier to draw the proposal in the \textit{Random Walk Phase} having the covariance matrix already decomposed.   

In the third phase random swaps between different chains are being performed according the rules described in Section \ref{stateDependentSwapStrategies}.
 
In the fourth phase the temperature scheme gets adapted. It's being done so as to have the ratio of differences between adjacent temperatures reflect the difference between the last steps probability of accepting the swap between the two levels and the theoretical value of $0.234$. The above-mentioned probability is calculated as if the swaps were drawn uniformly from the set of only neighbouring temperature levels alone. The rationale behind using the artificial acceptance rates stems from our requirement to make that decision independent of the number of temperature levels. In particular, the EE strategy is dependent upon the number of temperature levels (see Section~\ref{sec:numericalExperiments}). 

In the fifth phase we adapt the number of temperatures by cutting all the temperature levels with their corresponding square-root of the variance scaling factor being larger than the value of the optimal covariance of the proposal kernel.

\begin{algorithm}
\caption{One step of \emph{Adaptive Parallel Tempering}}
\begin{algorithmic}\label{algoDescription}
    \STATE \textbf{Input} $X^{[n]}, \Sigma^{[n]}, \theta^{[n]}, \mu^{[n]}, L^{[n]}, \beta^{[n]}$ 

    \STATE \,\,\,\,\emph{Random Walk Phase}
    \STATE $\ $
    \FOR {$\ell:=1$ to $L^{[n]}$ }
        \STATE $G \sim \mathcal{N}( 0, \Sigma^{[n]}_l )$
        \STATE $ Y_\ell := X_\ell^{[n]} + \exp{(\theta_\ell^{[n]})} \times G $
        \STATE $\log{\eta^{[n]}_\ell} := 0 \wedge  \beta_l^{[n]}
        \biggl( 
            \log\big(\pi(Y_\ell)\big) - \log\big(\pi(X_l^{[n]})\big) 
        \biggl)$  
        \IF{$\mathcal{U}(0,1) < \eta^{[n]}_\ell$} 
            \STATE $\tilde{X}^{[n]}_\ell := Y_\ell$
        \ELSE 
            \STATE $\tilde{X}^{[n]}_\ell := X^{[n]}_\ell $
        \ENDIF 
    \ENDFOR 
    \STATE \,\,\,\,\emph{Random Walk Adaptation}
    \STATE
    \FOR {$\ell:=1$ to $L^{[n]}$ } 
        \STATE\begin{align*}
            \mu^{[n+1]}_\ell &:= (1 - \gamma^{[n+1]}) \mu^{[n]}_\ell + \gamma^{[n+1]} \tilde{X}^{[n]}_\ell\\        
            \Sigma^{[n+1]}_\ell &:= (1 - \gamma^{[n+1]}) \Sigma^{[n]}_\ell \\ &\phantom{=}+ \gamma^{[n+1]} (\tilde{X}^{[n]}_\ell - \mu^{[n]}_\ell)(\tilde{X}^{[n]}_\ell - \mu^{[n]}_\ell)^\text{t}\\
            \theta^{[n+1]}_\ell &:= \theta^{[n]}_\ell + \gamma^{[n+1]} ( \eta^{[n]}_\ell - 0.234 )
        \end{align*} 
    \ENDFOR
     
    \STATE \,\,\,\,\emph{Random Swap Phase}
    \STATE
    \STATE Compute probabilities $\vec{p} := \left\{ 
        \pij( \tilde{X}^{[n]}_\ell ) 
    \right\}_{1 \leq i < j \leq L^{[n]}}$
    \STATE Sample $(i,j) \sim \vec{p} $
    \STATE $\alpha := 
    \frac{ \pij \big( \Tij( \tilde{X}^{[n]} ) \big) }{ \pij(\tilde{X}^{[n]}) } 
    \left( \frac{ \pi( \tilde{X}^{[n]}_i )}{ \pi( \tilde{X}^{[n]}_j ) } \right)^{ \beta^{[n]}_j-\beta^{[n]}_i }
    \wedge 1\;$
    \IF{$\mathcal{U}(0,1) < \alphaij$} 
        \STATE  $X^{[n+1]} := \Tij(\tilde{X}^{[n]} )$
    \ELSE 
        \STATE  $X^{[n+1]} := \tilde{X}^{[n]}$
    \ENDIF

    \STATE \,\,\,\,\emph{Temperature Scheme Adaptation}
    \STATE
    \STATE $T^{[n+1]}_1 = 1$
    \FOR{$\ell:=1$ to $L^{[n]} - 1$ }
        \STATE  
        \begin{align*}
            &\xi_\ell :=  \left( \frac{ \pi( X^{[n+1]}_\ell )}{ \pi( X^{[n+1]}_{\ell+1} ) } \right)^{ \beta^{[n]}_{\ell+1}-\beta^{[n]}_\ell } \wedge 1 \\
            &\log{(T^{[n+1]}_{\ell+1} - T^{[n+1]}_\ell)} := 
            \log{(T^{[n]}_{\ell+1} - T^{[n]}_\ell)} \\
            &\qquad\qquad\qquad\quad\qquad\qquad+ \gamma^{[n+1]} (\xi_\ell - 0.234)
        \end{align*}
    \ENDFOR
    \STATE $\beta^{[n+1]} := \frac{1}{T^{[n+1]}}$
    \STATE \,\,\,\,\emph{Number of Temperatures Adaptation}
    \STATE
    \IF{ $n > N_0$ }
        \STATE $L^{[n+1]} := \min \left\{ \ell \in \{1, \dots, L^{[n]} \}: 
        e^{\theta^{[n+1]}_\ell} \geq \frac{2.38}{\sqrt{d}}
      \right\}$
    \ELSE
        \STATE $L^{[n+1]} := L^{[n]}$
    \ENDIF    
\end{algorithmic}
\end{algorithm}


We now pass to questions regarding the theoretical underpinnings of the algorithm described above. The stability issues regarding parameters $\Sigma^{[n]}$, $\beta^{[n]}$, and $\mu^{[n]}$ are hard to analyse in general; thus, we proceed with examination thereof while restricting ourselves to the case where projections of parameters can be found altogether on some compact space. The precise description of the above-mentioned projections can be found in \citep{BM2}.

To obtain the Law of Large Numbers we need to replace \citep[Lemma 20]{supplement} by an analogous one which covers other swap strategies.
\begin{lemma}
    \label{lem:lipshitz}If swap strategy satisfies $\pij(x)=\pij(\Tij(x))$ for all $x$ and all $i<j$, and $\pij$ is Lipschitz continuous as a function of inverse
    temperatures $\beta$. 
    Then kernel $S$ is Lipschitz continuous in the $V$-norm with respect to the inverse temperature $\beta$.
\end{lemma}

\begin{proof}
    To explicitly denotes dependence upon $\beta$ we will use  notation $S_\beta$, $\alphaij(x,\beta)$, and $\pij(x,\beta)$ for $S$, $\alphaij$ and $\pij$ respectively.
     By definition of kernel $S_\beta$ for any measurable function $g$ and any $\beta,\beta'$ we have
     \begin{align*}
        \label{eq:S-V}
    &|S_{\beta}g(x) -S_{\beta'}g(x)|\\
    &=\Big|\sum_{i<j}\left(\pij(x,\beta)\alphaij(x,\beta)-\pij(x,\beta')\alphaij(x,\beta')\right)g(\Tij(x))\\
    &\phantom{=}+\sum_{i<j}\left(\pij(x,\beta')\alphaij(x,\beta')-\pij(x,\beta)\alphaij(x,\beta)\right)g(x)\Big|\\
    &\leq\sum_{i<j}\Big|\left(\pij(x,\beta)\alphaij(x,\beta)-\pij(x,\beta')\alphaij(x,\beta')\right)\Big| \bar{g}_{ij}(x),
    \end{align*}
    where $\bar{g}_{ij}(x)=\max\{|g(x)|,|g(\Tij)(x)|\}$.
    Since $\pij(x,\beta)\leq1$ and $\alphaij(x,\beta)\leq 1$, by triangle inequality we get
    \begin{multline*}
     \Big|\left(\pij(x,\beta)\alphaij(x,\beta)-\pij(x,\beta')\alphaij(x,\beta')\right)\Big|\\ \leq|\pij(x,\beta)-\pij(x,\beta')|+|\alphaij(x,\beta)-\alphaij(x,\beta')|
     \end{multline*}
    Terms whit differences of $\pij$ are Lipshitz by assumptions the Lipshitz conditions for terms with $\alphaij$ follows by the same arguments as in proof of Lemma 20 from
    \citep{supplement}
\end{proof} 

\goodbreak
\begin{theorem}
    Under assumptions of Lemma~\ref{lem:lipshitz} and regularity conditions as in Theorem 5 of \citep{BM2} for any function $f$ such that $\sup_x |f(x)\pi(x)^\tau|<\infty$, where $\tau<\frac{\beta_{\text{min}}}{2}$ with $\beta_{\text{min}}$ being the lower bound of the compact set of allowed temperatures, it is true that
     \[
     \frac{1}{n}\sum_{k=1}^nf(X_1^{[k]})\to \pi f \quad \text{a.s. ,}
     \]
    where $X_1^{[n]}$ is a first component of APT algorithm with fixed number of the temperatures levels.
\end{theorem}
\begin{proof}
    The proof can be carried out using the same arguments as the one to Theorem 5 of \citep{BM2}, with only one minor modification consisting in explicit use of the Lipschitz continuity of swap kernel given by Lemma~\ref{lem:lipshitz}.
\end{proof}

\begin{remark}
    It is true the number of temperatures is a non-increasing function of the iteration and hence almost every trajectory has asymptotically a constant number of levels. This fact suggests that the LLN might be satisfied with adaptation of the temperatures number. 
\end{remark}

\section{Numerical Experiments}\label{sec:numericalExperiments}

We have carried out computer simulations to test the functioning of our algorithms. In this section we propose two case studies of application of our algorithms together with a deepened analysis or the results obtained with them. Both of the provided examples are inherently multimodial.

Our first object of interest was testing the efficiency of our algorithm on an example with a controlled number of modes. An adequate toy-example used in literature as benchmark for testing MCMC algorithms comes from the article of \citet{liang-wong}. There, a mixture of 20 normal peaks with equal weights is being considered. The variances of these modes are small enough to keep the the spikes well separated, as might be inspected in Figure~\ref{fig:20peaks}. 

\begin{figure}
\centering
\caption{Twenty Gaussian Peaks Toy Example}\label{fig:20peaks}
 \includegraphics[width=8cm]{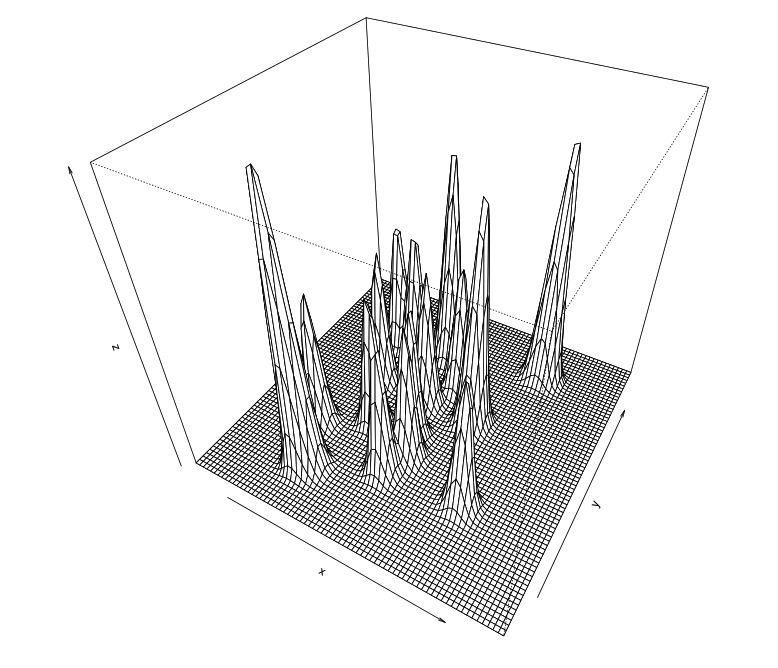}
\par
 \caption*{\small Observe that some peaks form more dense probability clusters that are still separated by large bands of almost zero probability regions.}
\end{figure}

\begin{table}[ht]
\centering
\caption{\textit{Root Mean Square Error} of estimators of the two first moments.}\label{tab:mseEE2d}
\begin{tabular}{lcccc}
  \toprule
RMSE & $\EV X_1$ & $\EV X_2$ & $\EV X_1^2$ & $\EV X_2^2$ \\ 
  \midrule
L=2 & 0.87 & 1.16 & 13.71 & 11.25 \\ 
 L=3 & 0.36 & 0.50 & 9.24 & 4.96 \\ 
  L=4 & 0.33 & 0.41 & 9.25 & 4.32 \\ 
  L=5 & 0.33 & 0.41 & 9.05 & 4.16 \\ 
  \bottomrule
\end{tabular}
\par
\bigskip
\caption*{\small Calculated for 2D Gaussian mixtures toy target using Equi-Energy type swap strategy for
a different number of temperature levels.}
\end{table}

We have performed simulations (500 runs) to check the quality of the algorithms. Since the actual first and second moments of this distribution are readily obtainable by direct calculation, we could evaluate the {\it Root Mean Square Error} of the moments estimators, as gathered in Table~\ref{tab:mseEE2d}, and compare ratios thereof while using different swap strategies, see Table~\ref{tab:ratio_rmse_2d}. Every simulation consisted of 2,500 iterations of burn-in followed by 5,000 steps of the main algorithm.
\begin{table*}[ht]
\centering
  \captionsetup{justification=centering}
  \caption{The efficiency of different algorithms in discovering modes of a 2D target distribution.}\label{tab:modes2d} 
\begin{tabular}{lcccccccccc}
  \toprule
&& \multicolumn{3}{c}{$L=3$} & \multicolumn{3}{c}{$L=4$} & \multicolumn{3}{c}{$L=5$} \\
\cmidrule(r){3-5}\cmidrule(r){6-8} \cmidrule(r){9-11}
Simulations' Feature&$L = 2$&EE&AL&RA&EE&AL&RA&EE&AL&RA\\
\midrule
  No missing modes (\%)   & 62.8  & 99.8  & 99.6  & 99.8  & 100 & 99.2  & 99.4  & 100 & 99.4  & 98.8  \\ 
b  Mean Absolute Error     & 0.59  & 0.32  & 0.34  & 0.35  & 0.3 & 0.36  & 0.37  & 0.29& 0.38  & 0.4   \\ 
  \bottomrule
\end{tabular}
\par
\bigskip
\caption*{\small Summary of the results of the \PT\, algorithm applied to the 2D Gaussian mixtures toy target with number of temperature levels equal to $L \in \{ 2,3,4,5\}$. Three different swap strategies were tested:
Equi Energy type moves (EE), swapping only two adjacent levels (AL), and swapping between all levels at random (RA). The {\it Mean Absolute Error} of time spent in each mode is approximated by calculating $\frac{1}{20}\sum_{i=1}^{20}\frac{|t_i-0.05|}{0.05}$, $t_i$ being the proportion of time
spent in $i^\text{th}$ mode.}
\end{table*}

\begin{table}[ht]
\centering
\caption{Ratio of RMSE of estimators of the two first moments}\label{tab:ratio_rmse_2d}
\begin{tabular}{lccccc}
  \toprule
Temperatures $\text{N}^\text{o}$  & Ratio & $\EV X_1$ & $\EV X_2$ & $\EV X_1^2$ & $\EV X_2^2$ \\ 
  \midrule
  \multirow{2}{*}{$L = 3$}&  AL/EE     & 1.04   & 0.96  & 1.01  & 0.97 \\ 
                          &  RA/EE     & 1.04   & 0.97  & 1.05  & 1.00 \\ 
  \cmidrule(r){1-6}
  \multirow{2}{*}{$L = 4$}&  AL/EE     & 1.13   & 1.27  & 1.00  & 1.21 \\ 
                          &  RA/EE     & 1.23   & 1.20  & 1.05  & 1.18 \\ 
  \cmidrule(r){1-6}                         
  \multirow{2}{*}{$L = 5$}&  AL/EE     & 1.30   & 1.31  & 1.06  & 1.27 \\ 
                          &  RA/EE     & 1.24   & 1.30  & 1.06  & 1.30 \\ 
  \bottomrule
\end{tabular}
\par
\bigskip
\caption*{ \small Calculated for 2D Gaussian mixtures target for different number of temperatures and different choice of swap strategy. Confront Table~\ref{tab:modes2d} for swap strategy encodings.} 
\end{table}
We have also modified the upper mentioned example by considering the product of that distribution with six independent uniform distributions. The final distribution is composed therefore of 20 peaks that are being stretched to additional dimensions.  The purpose of that experiment was to test how our algorithms cope with the difficulty of exploring a highly multidimensional space. Results of numerical experiments (100 runs) are contained in Table~\ref{tab:modes8d}. Each run was preliminated by 5,000 steps of burn-in, followed by 10,000 steps of the proper algorithm. 

\begin{figure}
{\centering
\caption{Twenty Gaussian Peaks 2D Simulation}\label{fig:lev2d}
 \includegraphics[width=8cm]{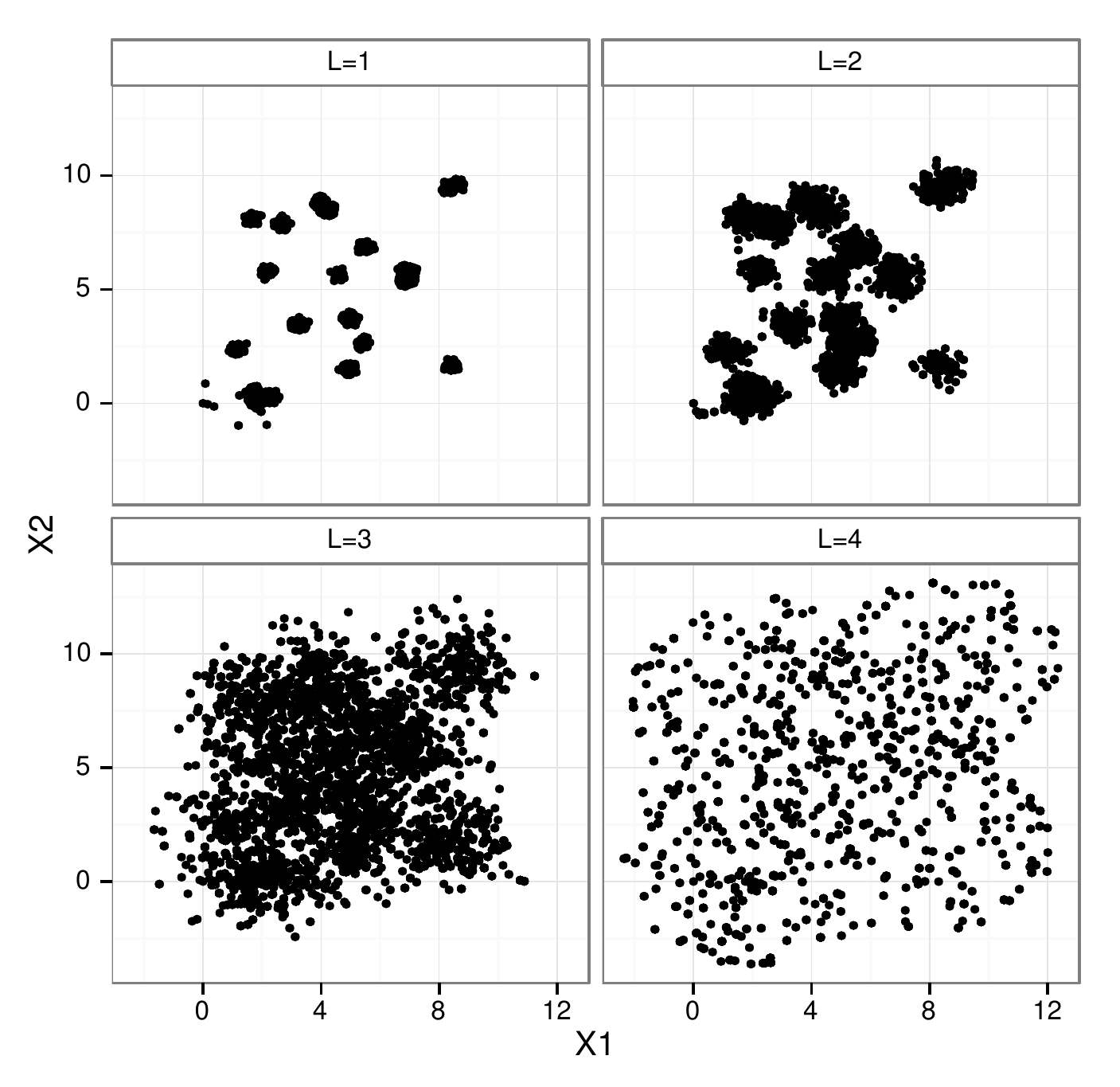}}
 \par
\caption*{\small Here we plot one individual trajectory. Each subplot represents the natural projection of that trajectory on the state space containing the original multimodial distributions and their transforms. The simulation lasted 7500 iterations, preceded by 2500 iterations of burn-in. The criterion for temperature scheme reduction was not set on until the burn-in period has finished. When turned on, it reduced the problem to three levels out of original four. } 
\end{figure}

\begin{figure}
\centering
\caption{Twenty Gaussian Cylinders 8D Simulation}\label{fig:lev8d}
 \includegraphics[width=8cm]{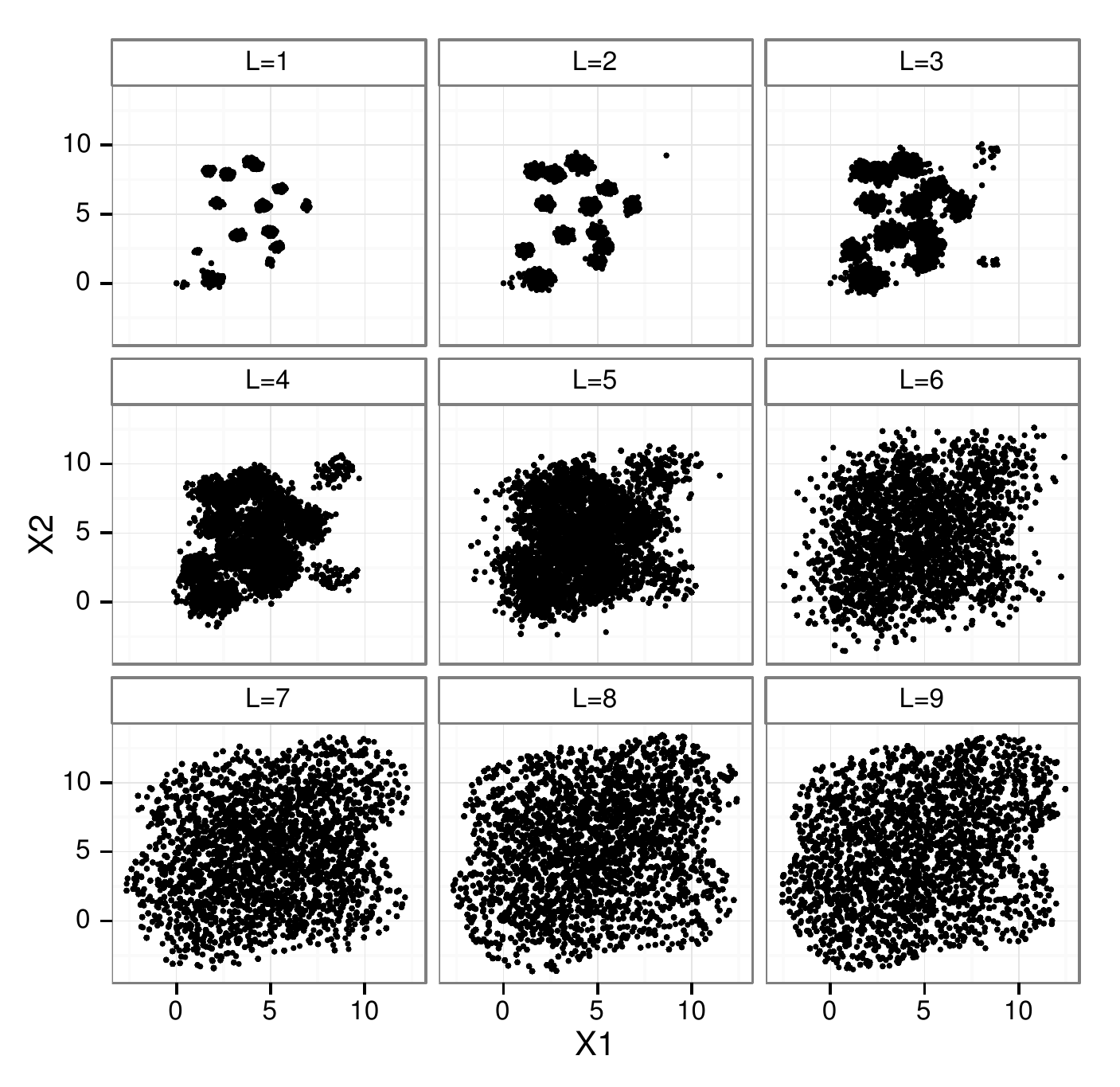}
 \par
 \caption*{\small Here we plot one individual trajectory. Each subplot represents the natural projection of that trajectory on the state space containing the original multimodial distributions and their transforms. 

 The simulation lasted 15,000 iterations, preceded by 5,000 iterations of burn-in. The criterion for temperature scheme reduction was not set on until the burn-in period has finished. When turned on, it reduced the temperature scheme usually to four or five levels. When the criterion has been turned on after 7,500 iterations, it indicated always that one should consider five levels.}
\end{figure}

To see how an exemplary trajectory manages its search for probability modes confront Figures~\ref{fig:lev2d} and \ref{fig:lev8d}. It can be clearly seen that not at all temperature levels can the algorithm easily travel between different modes.

It was of interest to us how well do the algorithms perform when dealing with the inherent multimodiality of the upper-mentioned problems. In case of two first problems mentioned above we do know the number of modes a priori and checking whether an algorithm managed to explore a particular mode could be dealt with use of clusterings of the state space into regions representing particular modes. Tables \ref{tab:modes2d} and \ref{tab:modes8d} do present the efficiency of estimates as measured by the expected number missed modes during one run of the algorithm and also the frequency of case when the algorithms does not miss any mode.
\begin{table*}[ht]
\centering
\caption{The efficiency of different algorithms in discovering modes of a 8D target distribution.}\label{tab:modes8d}
\begin{tabular}{lcccccccccccc}
  \toprule
  &\multicolumn{3}{c}{$L=3$} & \multicolumn{3}{c}{$L=5$} & \multicolumn{3}{c}{$L=7$} & \multicolumn{3}{c}{$L=9$}\\
  \cmidrule(r){2-4}\cmidrule(r){5-7} \cmidrule(r){8-10} \cmidrule(r){11-13} 
 Simulations' Feature&EE&AL&RA&EE&AL&RA&EE&AL&RA&EE&AL&RA\\
 \midrule
 No missing modes (\%) & 0 & 0 & 0 & 17.3 & 16.3 & 5.1 & 17.3 & 11.2 & 3.06 & 19.4 & 6.12 & 0 \\ 
  Average number of missing modes & 16 & 16 & 16 & 1.6 & 1.6 & 2.7 & 1.8 & 2.2 & 3.4 & 1.6 & 2.4 & 5.2 \\ 
  Mean Absolute Error & 1.7 & 1.7 & 1.7 & 0.78 & 0.79 & 0.89 & 0.79 & 0.86 & 0.93 & 0.8 & 0.85 & 1.1 \\ 
   \bottomrule
\end{tabular}
\par
\bigskip
\caption*{\small Summary of the results of the \PT\, algorithm applied to the 8D Gaussian mixtures toy target with number of temperature levels equal to $L \in \{ 3,5,7,9\}$. Confront Table~\ref{tab:modes2d} for swap strategies' encodings and definition of the \textit{Mean Absolute Error}.}

\end{table*}

We have run the algorithms in two modes. In first mode we have neglected temperature scheme reduction (but maintained the temperature adaptation as such). These simulations were carried out to provide insight into how the initial number of temperatures affects different strategies results. The results for the 2D toy example are gathered in \mbox{Tables \ref{tab:modes2d}-\ref{tab:ratio_rmse_2d}}. It can be generally observed that the \textit{Equi Energy} (EE) strategy slightly outperforms the state-independent swap strategies, the difference getting bigger with more temperature levels being taken into account. Clearly this can be attributed to the fact that a state-dependent swap is more immune to the quadratic growth in the number of possible swaps: choosing a swap uniformly from all potential swaps renders exchanges between the base level and other levels less likely in overall. On the other hand, choosing only adjacent levels (AL) seems to be only slightly worse. 

We have therefore tumbled upon an important rationale for temperature scheme adaptation as such: if the probability of proposed swaps is not concentrated on a relatively small subset, then it is less likely for the algorithm to convey information about a new mode discovery to the base chain. Observe however, that one cannot \textit{a priori} underestimate the proper number of temperatures; for in such case the tempered chains could be still exploring poorly separated probability modes of the transformated target distribution. For instance, in Figure~\ref{fig:moments} one can explicitly see that there is a huge difference between the errors in first two moments estimates for the 8D toy example when one uses 3 temperature levels instead of 4 or more. Similarly, in Table~\ref{tab:modes2d} one notices that AL and RA strategies do perform worse with more temperature levels considered, for on average the number of modes that they miss slightly grows. It is completely different in case of the state-dependent EE strategy that works actually better with the growing number of temperatures. The supremacy of this state-dependent strategy can be clearly seen in Table~\ref{tab:ratio_rmse_2d} that presents direct comparison of the \textit{Root Mean Square Errors} of other strategies when compared to that of the EE strategy. \mbox{Tables \ref{tab:modes8d}-\ref{tab:ratio_rmse_8d}} show similar results for the 8D toy example. 
\begin{table*}[ht]
\centering
\caption{Ratio of the \textit{Root Mean Square Error} of estimators of the two first moments.}\label{tab:ratio_rmse_8d}
\begin{tabular}{lcccccccccccccc}
  \toprule
 & \multicolumn{2}{c}{L=3}&\multicolumn{2}{c}{L=4}&\multicolumn{2}{c}{L=5}&\multicolumn{2}{c}{L=6}&\multicolumn{2}{c}{L=7}&\multicolumn{2}{c}{L=8}
 &\multicolumn{2}{c}{L=9}\\ 
\cmidrule(r){2-3}\cmidrule(r){4-5}\cmidrule(r){6-7}\cmidrule(r){8-9}\cmidrule(r){10-11}\cmidrule(r){12-13}\cmidrule(r){14-15}
Moment&AL/EE&RA/EE&AL/EE&RA/EE&AL/EE&RA/EE&AL/EE&RA/EE&AL/EE&RA/EE&AL/EE&RA/EE&AL/EE&RA/EE\\
  \midrule
  $\EV X_1$ & 1.00 & 1.00 & 1.08 & 1.02 & 0.97 & 1.03 & 1.03 & 0.94 & 1.33 & 1.38 & 0.95 & 1.14 & 1.03 & 1.16 \\ 
  $\EV X_2$ & 1.01 & 1.02 & 1.07 & 1.00 & 1.15 & 0.98 & 0.97 & 1.06 & 1.12 & 1.11 & 0.95 & 1.35 & 1.13 & 1.51 \\ 
  $\EV X_1^2$ & 1.00 & 1.00 & 1.02 & 1.01 & 0.93 & 1.05 & 0.96 & 0.91 & 1.04 & 1.02 & 0.96 & 1.08 & 1.01 & 1.08 \\ 
  $\EV X_2^2$ & 1.01 & 1.01 & 0.95 & 1.00 & 1.09 & 0.97 & 1.04 & 1.11 & 1.15 & 1.24 & 1.01 & 1.43 & 1.05 & 1.55 \\ 
  \bottomrule
\end{tabular}
\par
\bigskip
\caption*{\small Calculated for 8D Gaussian mixtures target for different number of temperatures with different swap strategies, at different temperature levels. Confront Table~\ref{tab:modes2d} for swap strategy encodings. } 
\end{table*}

Table~\ref{tab:acceptanceRates} shows also another interesting phenomenon: the overall percentage of accepted swaps grows with the number of temperature levels when applying the EE strategy. It suggests that the algorithm performs well in finding points in the state-space with similar probabilities. It also suggests that that truly the temperature adaptation procedure should not be using the original probabilities of acceptances but some other quantities independent of the number of initial temperatures. This provides a rationale for our choice of probabilities that would result from application of any swap proposal kernel that is symmetric (see the description of the fourth phase of the algorithm in Section \ref{algorithm}).       
\begin{table}[ht]
  \centering
  \caption{Average acceptance rate of swap steps.}\label{tab:acceptanceRates}
  \begin{tabular}{lccccccc}
    \toprule
  Strategy & L=3 & L=4 & L=5 & L=6 & L=7 & L=8 & L=9 \\ 
    \midrule
  EE & 0.38 & 0.42 & 0.45 & 0.45 & 0.46 & 0.46 & 0.46 \\ 
    RA & 0.16 & 0.12 & 0.10 & 0.08 & 0.07 & 0.06 & 0.05 \\ 
    \bottomrule
  \end{tabular}
  \par
  \bigskip
  \caption*{\small Calculated for 8D Gaussian mixtures target for Equi-Energy type swap strategy and random swap strategy  at different temperature levels Confront Table~\ref{tab:modes2d} for swap strategy encodings.}
\end{table}

The general message from the above analysis is clearly that there exists a threshold number of temperature levels that significantly ameliorates the \PT\, algorithm's performance and that the state-dependent strategy might partially solve problems resulting in unlikely swaps with the base level chain. We shall now describe the second phase of our numerical experiments with the toy examples, namely the verification of how well does our temperature scheme reduction criterion works in practice. 

To this end we have carried out yet another series of experiments (100 runs). In the 2D toy example (7,500 steps, out of which 2,500 burn-in) the ultimate number of temperature levels in every simulation reached three. In the 8D example the result depended on the burn-in period: if out of 15,000 steps the burn-in amounted to 7,500 then the algorithm always reached five levels of temperatures. When given a shorter burn-in period in $85\%$ of cases it reached five levels, otherwise descending even to four. It seems therefore that the algorithm approaches levels that could have been intuitively chosen and, what is important, in a more conservative way: it simply at worse slightly overestimates the actual level that might have been chosen by visual inspection of the "scree plots" such as the ones depicted in Figure~\ref{fig:moments}.  
\begin{figure}
\centering
\caption{Root mean square error of estimators of first and second moments}\label{fig:moments}
 \includegraphics[width=8cm]{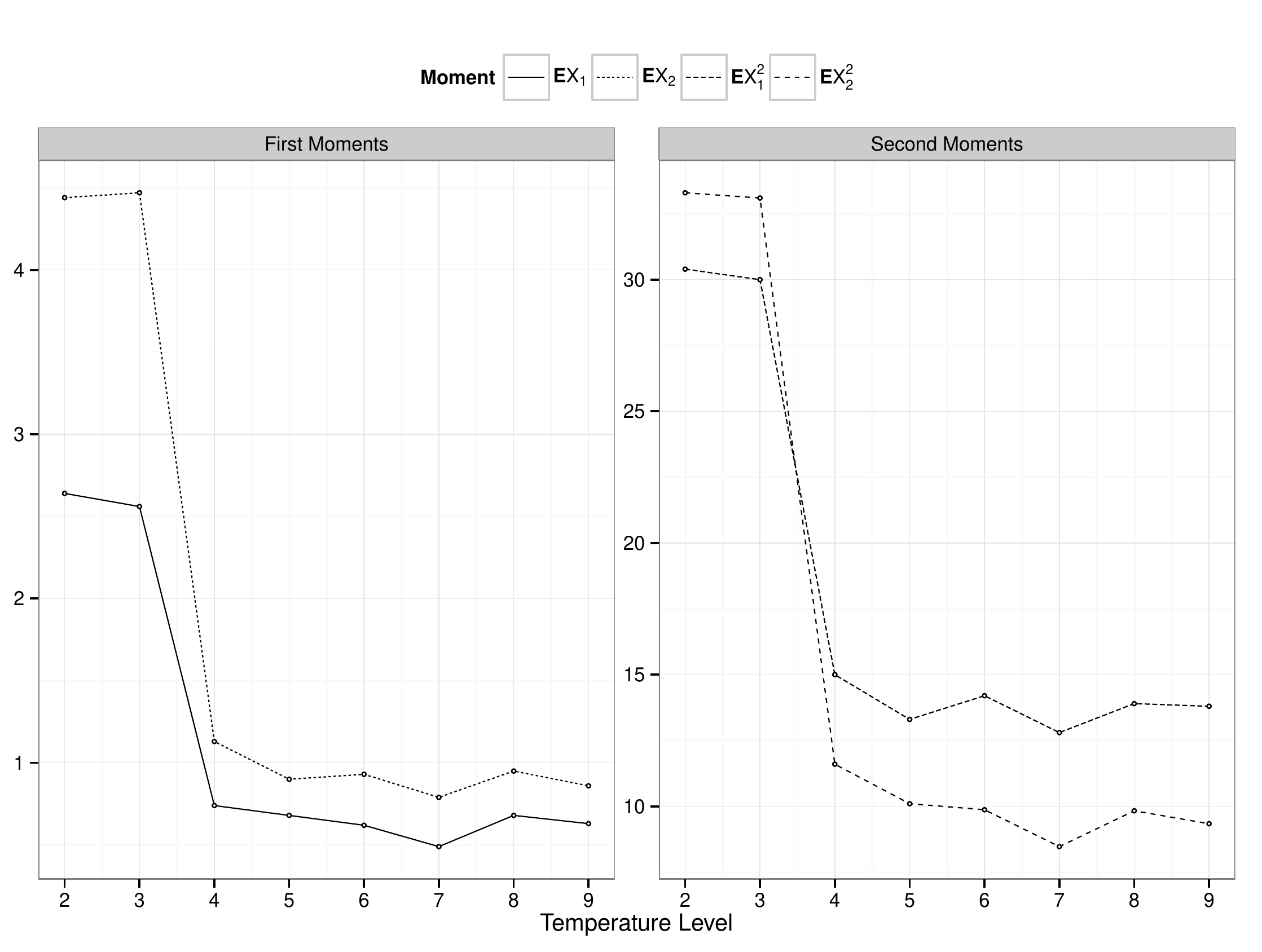}
 \par
 \caption*{\small Calculated for 8D Gaussian mixtures target for different number of temperatures with Equi-Energy type swap strategy at different temperature levels. Apparently there exists a threshold value of the initial number of temperatures that renders the estimates much more accurate. Not knowing it a priori, we are bound to overshot the true value to obtain high quality estimates. This is also a rationale behind any automated temperature reduction algorithm.}
\end{figure}

To show that the algorithms developed in this article can be of use in applications, we have tested them on the problem of the estimation of linear model coefficients. The formulation of the problem comes from \citet{park-casella} and is known in the literature under the name of \textit{Bayesian Bridge} regression. It consists in performing Bayesian inference on the parameters of a standard linear model
\begin{equation*}
  \bdy = \mu \boldsymbol{1}_n + \boldsymbol{X} \boldsymbol{\beta} + \boldsymbol{\epsilon},
\end{equation*}
where one postulates that $\boldsymbol{\epsilon} \sim \mathcal{N}(\boldsymbol{0}_n, \Sigma)$, resulting in 
\begin{equation*}    
  \bdy \sim \mathcal{N}(\boldsymbol{X} \boldsymbol{\beta}, \sigma \mathbb{I}_n),     
\end{equation*}      
and additionally one poses a shrinkage prior on the $\boldsymbol{\beta}$ coefficients 
\begin{equation*}    
  \pi(\boldsymbol{\beta}) = \prod_{j=1}^p \frac{\lambda}{2} e^{-\lambda |\beta_j|^q}.
\end{equation*}      
Here, $p$ is the number of different features gathered in the data set $X$, and $n$ is the overall number of sample points. Finally, one poses a scale invariant prior on the model variance $\sigma \sim \pi( \sigma ) = \frac{1}{\sigma^2}$. 
There is an explicit relation between the \textit{maximum a posteriori} estimate of parameters in the above mentioned model and the solution to the problem of coefficients estimates provided by the Lasso algorithm \citep{tibshirani-LASSO-seminal}. 
The \textit{a priori} independence of the $\boldsymbol{\beta}$ from $\sigma^2$ might render the $\boldsymbol{\beta}$ \textit{a posteriori} distribution multimodial, as in the case studied herein - following \citet{park-casella} we have run the algorithms on the \textsc{diabetes} dataset taken from \citet{efron-hastie-tibshirani}. 

In this example the starting point for the algorithms (namely - a pair $(\boldsymbol{\beta}^{[0]}, \sigma^{[0]})$) was set to be a simple OLS estimate for the coefficients $\hat{\boldsymbol{\beta}}$ together with the resulting mean residual sum of squares being a proxy for real variance, $\hat{\sigma}$. We have carried out one hundred simulations each consisting of 20,000 iterations of burn-in and 100,000 iterations of the actual algorithm, for different initial number of temperatures.  
\begin{table}[ht]
  \centering
  \caption{Posterior Means of selected Parameters of the \textit{Bayesian Bridge} regression}\label{tab:bayesianLassoParams}
  \begin{tabular}{lcccccc}
  \toprule
   & $\beta_3$ & $\beta_9$ & $\log(\sigma)$ \\ 
    \midrule
  RWM   & 0.228 (0.64)    &  0.165  (0.390)   & 8.69 (0.003)   \\
    EE  & 723   (2.60)     &  0.240  (0.150)   & 8.3  (0.0031)  \\ 
    AL  & 724   (2.60)     &  0.199  (0.050)   & 8.3  (0.0027)  \\
    RA  & 723   (4.40)     &  0.141  (0.067)  & 8.3  (0.0028)  \\
    \bottomrule
  \end{tabular}
  \par
  \bigskip
  \caption*{\small Summary of the results of the \textit{Bayesian Bridge} regression. It can be seen that the simple random walk Metropolis (RWM) strategy does not succeed in exploring all the modes. When it comes to other strategies, no significant differences can be spotted.}
\end{table}
\begin{figure}
\centering
\caption{Posterior Distribtuion of Coefficients in Bridge Regression}\label{fig:bayesianBridgeRegression}
 \includegraphics[width=8cm]{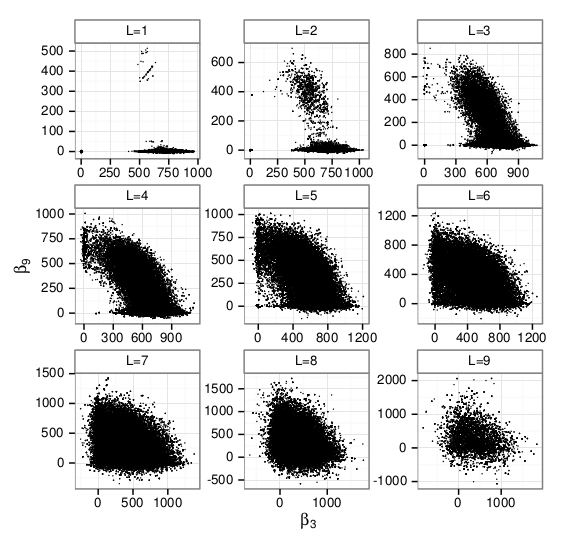}
\caption*{\small Above, projections of one simulated chain on the subspace spanned by two parameters in the Bayesian Lasso model. Visual inspection of base temperature ($L=1$) reveals the inherent multimodiality of the \textit{a posteriori} distribution. The simulations were performed on nine different temperature levels. One notices visually that at temperature ranging from the $6^\text{th}$ and $7^\text{th}$ level the modes are not easily recognisable, which implies that the corresponding chains no longer experience local behaviour. Also our criterion (cerified every 1000 steps of the algorithm) did cut off levels higher than 7.}
\end{figure}
The results of our simulations are represented in Table~\ref{tab:bayesianLassoParams}. While performing the simulations it was observed that a simple random walk could not leave the neighbourhood of $\boldsymbol{0}_n$ even though most of the probability mass lays elsewhere, as can be investigated from the plot for the first temperature level ($L = 1$) in Figure~\ref{fig:bayesianBridgeRegression}. 

Results on the final number of reached temperatures are gathered in Figure~\ref{fig:ultimateNumberOfTemperatureLevels} and demonstrate that its inclusion in the algorithm gives rather conservative results in terms of temperature levels reduction, which might be considered safe an option. 

\begin{figure}
\centering
\caption{Distribution of Ultimate Number of Temperature Levels in the \textit{Bayesian Bridge} Regression }\label{fig:ultimateNumberOfTemperatureLevels}
 \includegraphics[width=8cm]{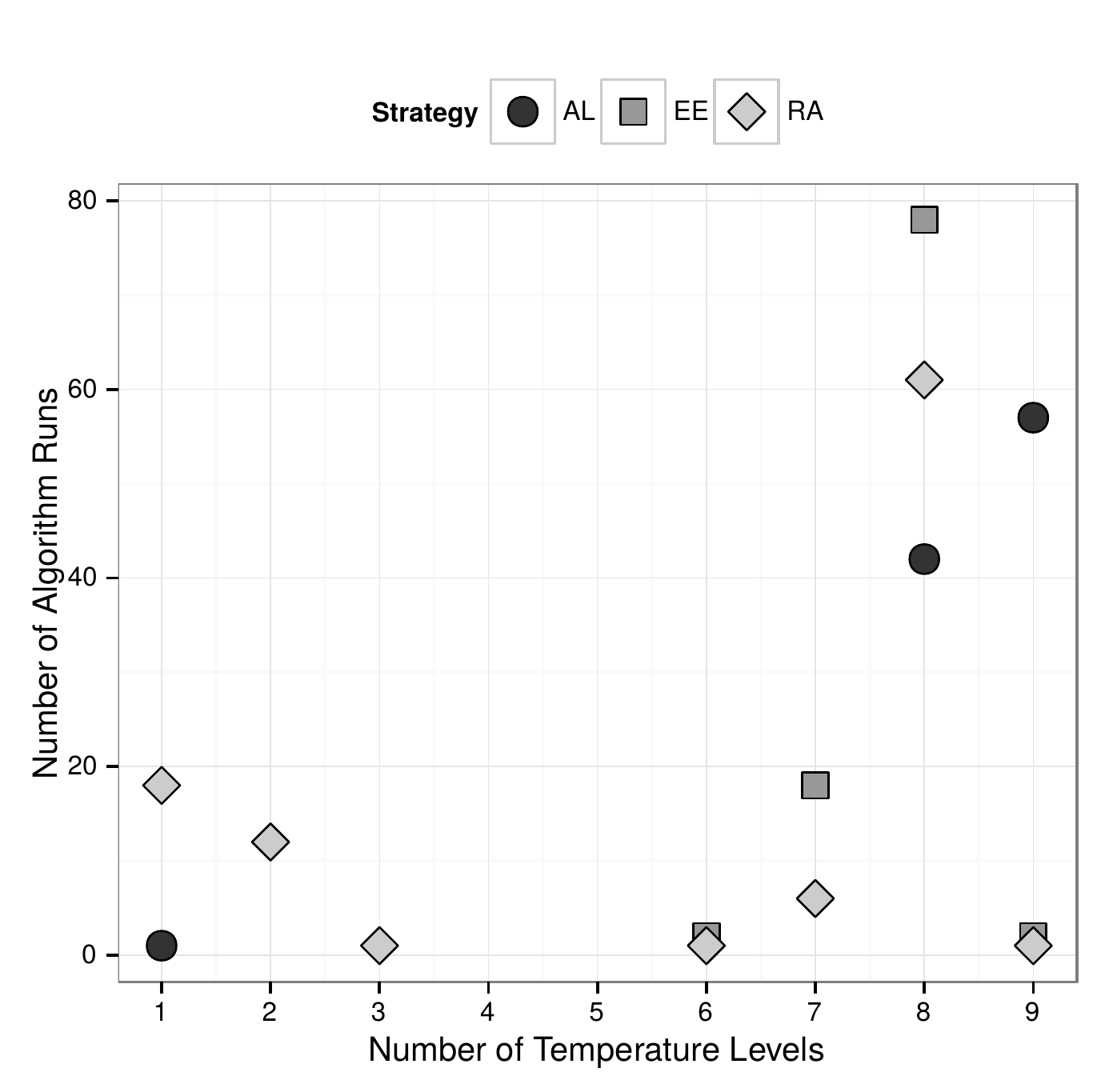}
 \par
\caption*{\small The above plot counts how many times different numbers of temperature were reached when using the temperature number adaptation and different swap strategies. The simulations lasted 100K iterations, preceded with 20K steps of burn-in. It is interesting to observe that the EE strategy never reached a number smaller than six and usually promoted a conservative choice of eight levels of temperature. Confront that with Figure~\ref{fig:lev8d}.} 
\end{figure}

\section{Conclusions}\label{conclusions}

In this article it was demonstrated that the \PT\, algorithm's efficiency is highly dependent upon the number of initial temperature levels; for underestimating that number leads to poor quality estimates resulting from neglecting some of the modes of the distribution of interest. On the other hand, overestimation of the number of temperatures does not significantly improve these estimates but highly increases the computational cost. Finally, if in that case one uses a naive state-independent strategy, then the overall efficiency of the algorithm drops due to longer waiting times for significant swaps, i.e. ones involving the base level temperature.

We have developed an algorithm with adaptable temperature scheme. Using it, one can overshot the number of temperature levels needed for performing good quality simulation and the procedure will automatically reduce some of the redundant levels. The results of our simulation support our claim of the method correctness. 

We have also proposed a novel for state dependent strategies that allow promotion of swaps based on various properties of the current state, e.g. equi-energy type moves, big-jumps promotion, and similar. We have shown that this framework is susceptible to analytical analysis based on already known results. Specifically, we have concentrated here on evaluating the equi-energy type strategy and thoroughly tested it. The results show that it is more efficient than the standard state-independent swap strategy, since the overall acceptance of a global moves is higher, the difference getting larger with the initial number of temperatures. Similar conclusions can be drawn on the behaviour of the \textit{Root Mean Square Error}.    

When it comes to future research, we plan to prepare an \textbf{R}-package containing all the above methods. 

All the scripts used in simulations are readily obtainable on request.


 \begin{acknowledgements}
We are grateful to Matti Vihola for useful discussions.
 \end{acknowledgements}

\bibliographystyle{spbasic}      
\bibliographystyle{spbasic}      


\end{document}